\documentclass[runningheads]{llncs}
\usepackage[T1]{fontenc}
\usepackage{graphicx}
\usepackage{amsmath}
\usepackage{enumitem}
\usepackage{amsfonts}
\usepackage{booktabs}
\usepackage{xcolor}
\usepackage{hyperref}
\usepackage{cleveref}

\usepackage{colortbl}

\usepackage{placeins}
\usepackage{float}
\usepackage[normalem]{ulem}

\usepackage{tikz}
\usetikzlibrary{tikzmark}

\usepackage{color}

\begin{document}
\title{Axiomatizing approximate inclusion}
%
%
\author{Matilda Häggblom \orcidID{0009-0003-6289-7853}}
\authorrunning{M. Häggblom}
%
\institute{University of Helsinki, Finland \\
\email{}}
\maketitle              
\begin{abstract}
We introduce two approximate variants of inclusion dependencies and examine the axiomatization and computational complexity of their implication problems. The approximate variants allow for some imperfection in the database and differ in how this degree is measured. One considers the error relative to the database size, while the other applies a fixed threshold independent of size. We obtain complete axiomatizations for both under some arity restrictions. In particular, restricted to unary inclusion dependencies, the implication problem for each approximate variant is decidable in PTIME. We formalise the results using team semantics, where a team corresponds to a uni-relational database. 

\end{abstract}
\section{Introduction}

In practice, databases often contain errors, and dependency relations such as functional and inclusion dependencies might not hold absolutely. It is therefore meaningful to consider their approximate variants. We define two approximate inclusion dependencies and study the axiomatizations and computational complexity of their implication problems. 

We present the results using team semantics \cite{hodges1997,hodges19972,vaananen2007}, where a team is a set of assignments corresponding to a uni-relation database. Inclusion dependencies are one of the most used dependency classes, and were adapted to team semantics in \cite{Galliani2014} as inclusion atoms $x\subseteq y$, describing that the values of $x$ are among the values of $y$. The results are thus immediately transferable to the database setting.

Dependence and exclusion atoms are both downward-closed, meaning that if a team satisfies an atom, so do its subteams. Their approximate versions were defined and axiomatized in \cite{Vaananen2017,haggblom2024}. The notion of approximation for downward-closed atoms naturally corresponds to finding a large enough subteam that satisfies the atom. For non-downward-closed atoms, the notion of approximation is less obvious. We extend the research program on approximate dependencies and their proof systems to non-downward-closed atoms. 


Inclusion atoms are not downward-closed, but union-closed: if an atom is satisfied by multiple teams, it is satisfied by their union. This begs the question of what an appropriate notion of approximation is, since notably, inclusion atoms are not upward-closed, i.e., extending a team can make the inclusion fail. We therefore propose that the approximate version of an inclusion atom $x\subseteq y$ should allow some number of values for $x$ to be missing from $y$. We introduce \emph{quantity approximate inclusion} atoms $x\subseteq_n y$, where $n$ values are allowed to be missing, and \emph{ratio approximate inclusion} atoms $x\subseteq_p y$, where $p$ is a rational number corresponding to a number of values relative to the size of the team.

Let us compare the two definitions with an example. Consider a team with 40 assignments over the variables $x:=FRESHMEN$ and $y:=REGISTERED$, listing $40$ freshmen of whom $10$ are not yet registered. Now the team satisfies the quantity and ratio approximate inclusion atoms $x\subseteq_{10} y$ and $x\subseteq_{\frac{1}{4} }y$. In practice, this information can be used to determine if a reminder to register needs to be sent to the students.


We consider the implication problem for each approximate variant over finite teams: does a set of atoms imply a given atom? For quantity approximate inclusion, we extend the system for inclusion atoms in \cite{CASANOVA198429} to a complete proof system for finite assumption sets in which the atoms' arities do not exceed that of the conclusion atom. Additionally, we establish that the decision problem remains PSPACE-complete, like for usual inclusion atoms \cite{CASANOVA198429}.

Similarly, we obtain the proof system for ratio approximate inclusion atoms by essentially adding the approximations to the system in \cite{CASANOVA198429}. It is not trivial that such an approach produces sound rules when the approximations are dependent on the size of the team, as we discuss in \Cref{conclusion}. With the carefully chosen definition, we show completeness for finite sets of unary ratio approximate inclusion atoms, and that their decision problems can be solved in PTIME. Unary inclusion atoms are often favoured in practice, due to better complexity results (see, e.g., \cite{cosmadakis}), which can now be extended to their approximate counterparts.  

%



%
%
\vspace{-.25cm}
\section{Preliminaries}
A team $T$ is a set of assignments $s:\mathcal{V}\longrightarrow M$, where $\mathcal{V}$ is a set of variables and $M$ is a set of values. In this paper, we assume that teams are finite. We write 
$x,y,\dots$ for finite (possibly empty) sequences of variables, with $|x|, |y|,\dots$ denoting their lengths. 
 For $x=x_1\dots x_n$ and $y=y_1\dots y_m$,  $xy$ denotes the sequence $x_1\dots x_n y_1\dots y_m$. We write $s(x)$ as shorthand for the tuple $ s(x_1)\dots s(x_n)$. 

Inclusion atoms are of the form $x\subseteq y$, where $|x|=|y|$ is the arity of the atom. We assume that no variable is repeated within the respective sequences $x$ and $y$.
We recall the semantics of inclusion atoms,
\begin{equation*}
     T\models x\subseteq y \text{ if and only if for all } s\in T \text{ there exists } s^\prime\in T : s(x)=s^\prime(y). 
\end{equation*}
Equivalently, we have that $T\models x\subseteq y$ if and only if $T[x]\subseteq T[y]$, where the value-sets are defined by $T[z]:=\{s(z)\mid s\in T\}$. 

It follows from the definition that inclusion atoms are union-closed: if $T_i\models x\subseteq y$ for each $i$ in a nonempty index set $I$, then $\bigcup_{i\in I}T_i\models x\subseteq y$, and have the empty team property: all inclusion atoms are satisfied by the empty team.

For a set of atoms $\Sigma$, we write $\Sigma\models x \subseteq y$ if for all teams $T$ such that $T\models u\subseteq v$ for all $u\subseteq v\in\Sigma$, also $T\models x\subseteq y$. 


\section{Quantity approximate inclusion}
We introduce \emph{quantity approximate inclusion atoms} $x\subseteq_n y$ that allow at most $n$ values of $x$ to be missing from $y$, independent of the size of the team.
\begin{definition}
    For $n\in\mathbb{N}$,  $T\models x\subseteq_n y$ if and only if $|T[x]\setminus T[y]|\leq n$. 
\end{definition}

Note that atoms with approximation $0$ coincide with the usual inclusion atoms. It also follows from the definition that quantity approximate inclusion atoms have the \emph{weak union closure property}: If $T\models x\subseteq_n y$ and $T'\models x\subseteq_m y$ then $T\cup T'\models x\subseteq_{n+m} y$. Next, we define a sound set of rules for the atoms. 
%
%
\begin{definition}
The rules for quantity approximate inclusion atoms are $Q1$-$Q5$.
\begin{enumerate}[align=left]
\item[$(Q1)$] $x\subseteq_0 x$.

\item[$(Q2)$] If $x\subseteq_n z$ and $z\subseteq_m y$, then $x\subseteq_{n+m} y$.

\item[$(Q3)$] If $xyz\subseteq_n uvw$, then $xzy\subseteq_n uwv$.$^*$ \hfill $^*$$|x|=|u|$ and $|y|=|v|$  

\item[$(Q4)$] If $xy\subseteq_n uv$, then $x\subseteq_n u$. 

\item[$(Q5)$] For $m\geq n$, if $x\subseteq_n y$, then $x\subseteq_m y$. 
\end{enumerate}  
\end{definition}

\begin{lemma}[Soundness]
The rules $Q1$-$Q5$ are sound.   
\end{lemma}
\begin{proof}
We give the proof for Q2. Let $T\models x\subseteq_n z$ and $T\models z\subseteq_m y$. Then $|T[x]\setminus T[z]|\leq n$ and  $|T[z]\setminus T[y]|\leq m$. Now $|T[x]\setminus T[y]|\leq |T[x]\setminus T[z]|+|T[z]\setminus T[y]|\leq n+m$, and we conclude $T\models x\subseteq_{n+m} y$.
\end{proof}

We show completeness under the assumption that the arities of the atoms in the assumption set do not exceed the arity of the conclusion. Hence, $Q4$ can be omitted. We construct counterexample subteams that each uses two unique values, as in the completeness proof for usual inclusion atoms in \cite{haggblom2025}. The union of the counterexample subteams then forms the final counterexample team. One such instance is illustrated in \Cref{Q team}.


\begin{theorem}[Completeness]
Let $\Sigma\cup\{x\subseteq_n y\}$ be a finite set of quantity approximate inclusion atoms with the arities of the atoms in $\Sigma$ being at most that of $l:=|x|$. If $\Sigma\models x\subseteq_n y$, then $\Sigma\vdash x\subseteq_n y$ using rules from $\{Q1, Q2, Q3,Q5\}$.\label{compl ratio}
\end{theorem}

\begin{proof}
%
Suppose that $\Sigma\not\vdash x\subseteq_n y$, then by $Q1$ and $Q5$, $x\neq y$. Let $i^l$ be the tuple $i\dots i$ of length $l$. We define the counterexample team $T=T_1\cup\dots\cup T_{n+1}$ by letting $s\in T_i$ if and only if the following conditions are satisfied:
\begin{enumerate}[label=(\arabic*)]
    \item  $s:V\longrightarrow \{i,i+\frac{1}{2}\}$.\label{cond 1}
    \item For the smallest $m$ such that $\Sigma\vdash x\subseteq_m w$, if $i\leq m$ then $s(w)\neq i^l$.
%
    \label{cond 3}
\end{enumerate}

Note that a team $T_i$ satisfying only condition \ref{cond 1} satisfies all possible inclusion atoms, while condition \ref{cond 3} is a minimal intervention to make it a counterexample team. See \cite{haggblom2025} for details. Further, in each subteam $T_i$, no tuples are removed from variable sequences shorter than $l$, and at most one tuple, namely $i^l$, is removed from sequences of length $l$. However, for longer sequences $w$, the number of tuples removed can be $2^{|w|-l}$ (e.g. both $i^li$ and $i^l(i+\frac{1}{2})$), forcing the arity restriction.

By the assumption $\Sigma\not\vdash x\subseteq_n y$ and condition \ref{cond 3}, the value $i^l$ is missing from $y$ in each subteam, hence $|T[x]\setminus T[y]|=n+1$ and $\Sigma\not\models x\subseteq_n y$.
It remains to show that for all $u\subseteq_k v\in\Sigma$, $T\models u\subseteq_k v$. The arguments are mostly analogous to the proof in \cite{haggblom2025}, and we include the only case specific to the approximations.


Let $u\subseteq_k v\in\Sigma$ and let $m$ be the smallest approximation for which $\Sigma \vdash x\subseteq_m u$. By the construction of the team, $|T[x]\setminus T[u]|=m$ and by $Q2$, $\Sigma\vdash x\subseteq_{m+k} v$, so $|T[x]\setminus T[v]|\leq m+k$. Now $|T[u]\setminus T[v]|\leq (m+k)-m=k$, hence $T\models u\subseteq_k v$.  

%
%
%
\qed\end{proof}
%
%
%
\vspace{-.52cm} \begin{table}[H]
     \centering
     \caption{Illustration of the counterexample team $T=T_1\cup T_2 \cup T_3$ for $\Sigma\not\vdash x_1x_2\subseteq_{2}y_1y_2$, where $\Sigma=\{x_1x_2\subseteq_{2}w_1w_2,\,
w_1w_2\subseteq_{1}y_1y_2\}$. In the subteams below, any assignment that includes the listed tuple in a dashed row is not in $T_i$, while all remaining assignments to the values $\{i,i+\frac{1}{2}\}$ are in $T_i$ when $i\in\{1,2,3\}$.}
     \label{Q team} \vspace{-.5cm}
     \parbox{.33\textwidth}{
\centering
 \[
\begin{array}{c c cc cc }
\toprule 	
T_1:\, &\,x_1\, x_2 & & w_1\,  w_2 & &y_1\,  y_2\,\\ 
\midrule
 \quad\quad\,\,\tikzmark{start1}& \,&  & 1\,\,\, 1 & &  \,\,\quad\quad\,\,\tikzmark{end1}\\
\quad\quad\,\,\tikzmark{start2}&  \, & &&  & 1\,\,\,1\,
\quad\tikzmark{end2}  \\
 & \,1\,\,\,  1& & \frac{3}{2}\,\,\, 1 & &\dots  \,
 \vspace{5pt}   \\ \bottomrule
\end{array}
\] }
\parbox{.32\textwidth}{
\centering
 \[
\begin{array}{c c cc cc }
\toprule 	
T_2:\, &\,x_1\, x_2 & & w_1\,  w_2 & &y_1\,  y_2\,\\  
\midrule 
\quad\quad\,\,\tikzmark{start3}  & \,&  & 2\,\,\, 2 & & \,\,\quad\quad\,\,\tikzmark{end3}\\
\quad\quad\,\,\tikzmark{start4} & \, & &&  & 2\,\,\,2\,
\quad\tikzmark{end4}
  \\
 & \,{2}\,\,\, {2}& & 2\,\,\, \frac{5}{2} & &\dots  \,
 \vspace{5pt}   \\\bottomrule
\end{array}
\] }
\parbox{.33\textwidth}{
\centering
 \[
 \begin{array}{c c cc cc }
\toprule 	
T_3:\, &\,x_1\, x_2 & & w_1\,  w_2 & &y_1\,  y_2\,\\  
\midrule 
 & \,  \,\,\,  & & 3\,\,\, 3 & & \, \\ 
\quad\quad\,\,\tikzmark{start5} & \, & &&  & 3\,\,\,3\,\quad\tikzmark{end5}
  \\
 & \, 3 \,\,\,  3& & \frac{7}{2}\,\,\, \frac{7}{2} & &\dots  \,
  \vspace{5pt}  \\ 
  \bottomrule
\end{array}
\]}\vspace{-.44cm}
 \end{table}

\tikz[remember picture] \draw[overlay, dashed][dash pattern=on 4pt off 2pt] ([yshift=.35em]pic cs:start1) -- ([yshift=.35em]pic cs:end1);

  \tikz[remember picture] \draw[overlay, dashed][dash pattern=on 4pt off 2pt] ([yshift=.35em]pic cs:start2) -- ([yshift=.35em]pic cs:end2);

 \tikz[remember picture] \draw[overlay, dashed][dash pattern=on 4pt off 2pt] ([yshift=.35em]pic cs:start3) -- ([yshift=.35em]pic cs:end3);

  \tikz[remember picture] \draw[overlay, dashed][dash pattern=on 4pt off 2pt] ([yshift=.35em]pic cs:start4) -- ([yshift=.35em]pic cs:end4);

    \tikz[remember picture] \draw[overlay, dashed][dash pattern=on 4pt off 2pt] ([yshift=.35em]pic cs:start5) -- ([yshift=.35em]pic cs:end5);

    \vspace{-2.21cm}

The arity restriction in \Cref{compl ratio} is due to the counterexample team not necessarily satisfying atoms of larger arity than the conclusion $x\subseteq_n y$. However, we can still use information obtained from such atoms: given any set $\Sigma$, we can modify it to obtain an assumption set $\Sigma^\prime$ by the soundness of $Q4$ such that all atoms have arity at most $|x|$. Now, $\Sigma^\prime\models x\subseteq_n y$ if and only if there is a derivation using the rules $Q1$, $Q2$, $Q3$ and $Q5$. 

 Like for usual inclusion atoms, for finite $\Sigma$, deciding whether $\Sigma\models x\subseteq_n y$ is PSPACE-complete. In \Cref{unary complexity} of the next section, we show that it is in PTIME for unary atoms.

\begin{theorem}[Complexity]
Let $\Sigma\cup\{x\subseteq_n y\}$ be a finite set of quantity approximate inclusion atoms with the atoms in $\Sigma$ having arity at most $|x|$. Deciding whether  $\Sigma\models x\subseteq_n y$ is PSPACE-complete.
\end{theorem}
\begin{proof}
PSPACE-hardness follows from the result for inclusion atoms in \cite{CASANOVA198429}. Note that the arity restriction does not decrease the complexity, which stems from the permutation rule $Q3$. For membership, we adapt the proof in \cite{CASANOVA198429}:  guess a path (represented by the inclusion symbol) from $x$ to $y$ and record the approximation at each step. Such a path exists with a final approximation of at most $n$ if and only if $\Sigma\models x\subseteq_n y$.
\qed\end{proof}

\section{Ratio approximate inclusion}
We introduce \emph{ratio approximate inclusion atoms} $x\subseteq_p y$, where the permissible number of missing values depends on the size of the team. 

 
\begin{definition}\label{incl.def.alt.3}
For $p\in  [0,1]\cap\mathbb{Q}$, $T\models x\subseteq_p y$ if and only if $|T[x]\setminus T[y]|\leq p|T|$.
\end{definition}

The approximation $0$ coincides with the usual inclusion atom, and the approximation $1$ allows all values from $x$ to be missing from $y$. 

This section's focus is unary atoms $x\subseteq_p y$ with $|x|=1$, but we still present a more general sound proof system. 

\begin{definition}
The rules for ratio approximate inclusion atoms are $R1$-$R6$.
\begin{enumerate}[align=left]
\item[$(R1)$] $x\subseteq_0 x$.

\item[$(R2)$] If $x\subseteq_p z$ and $z\subseteq_q y$, then $x\subseteq_{p+q} y$.

\item[$(R3)$] If $xyz\subseteq_p uvw$, then $xzy\subseteq_p uwv$.$^*$ \hfill $^*$$|x|=|u|$ and $|y|=|v|$  

\item[$(R4)$] If $xy\subseteq_p uv$, then $x\subseteq_p u$. 

\item[$(R5)$] For $q\geq p$, if $x\subseteq_p y$, then $x\subseteq_q y$.   

\item[$(R6)$] $x\subseteq_1 y$.

\end{enumerate}  
\end{definition}

\begin{lemma}[Soundness]
The rules $R1$-$R6$ are sound.   
\end{lemma}
\begin{proof}
We show soundness for $R2$. 
Let $T\models x\subseteq_p z$ and $T\models z\subseteq_q y$. Then $|T[x]\setminus T[z]|\leq p|T|$ and $|T[z]\setminus T[y]|\leq q|T|$. Now, $|T[x]\setminus T[y]|\leq |T[x]\setminus T[z]| + |T[z]\setminus T[y]|\leq (p+q)|T|$, and we conclude $T \models x\subseteq_{p+q} y$.
%
%
%
%
 \qed\end{proof}

 We prove completeness for finite sets of unary atoms, allowing the omission of the rules $R3$ and $R4$. The proof is inspired by the corresponding one in \cite{Vaananen2017}. See \Cref{R team} for an instance of the counterexample team.
 

 \begin{theorem}[Completeness]
     Let $\Sigma\cup \{x\subseteq_p y\}$ be a finite set of unary inclusion atoms. If $\Sigma\models x\subseteq_p y$, then $\Sigma\vdash x\subseteq_p y$ using rules from $\{R1,R2,R5,R6\}$. \label{R compl}
 \end{theorem} 
 \begin{proof}
 Suppose that $\Sigma\not\vdash x\subseteq_p y$. By $R1$, $R5$ and $R6$, we can assume that $x\neq y$ and $p<1$. 
Let $n-1$ be the least common denominator for the approximations appearing in $\Sigma\cup \{x\subseteq_p y\}$. For each approximation $q$, let $q^\prime$ be such that $q=\frac{q'}{n-1}$. 

We build the counterexample team $T=\{s_1,\dots s_n\}$ according to the smallest approximation $q$ for which $\Sigma\vdash x\subseteq_{q} w$, by 
%
%
 \[   
s_i(w) = 
     \begin{cases}
       1 &\quad\text{if } i\leq q'+1,\\
       i &\quad \text{otherwise.} \\
     \end{cases}
\]

Since $\Sigma\vdash x\subseteq_0 x$ by $Q1$ and $\Sigma\not\vdash x\subseteq_p y$, we have that $|T[x]\setminus T[y]|\geq p'+1$. Now $|T[x]\setminus T[y]|/n\geq\frac{p'+1}{n} > \frac{p'}{n-1}=p$, thus $T\not\models x\subseteq_p y$. 

It remains to show that for $u\subseteq_r v\in\Sigma$, $T\models u\subseteq_q v$. 
Consider the smallest $q$ such that $\Sigma \vdash x\subseteq_q u$. By the construction of the team, $|T[x]\setminus T[u]|=q'$ and by $Q2$ we have that $\Sigma\vdash x\subseteq_{q+r} v$, hence $|T[x]\setminus T[v]|\leq q'+r'$. Now $|T[u]\setminus T[v]|\leq (q'+r')-q'=r'=r(n-1)< rn$, so $T\models u\subseteq_q v$.  
%
 \qed\end{proof}

 Restricted to unary atoms, each of the decision problems for the quantity- and ratio approximate inclusion atoms is in PTIME.
 
\vspace{-.35cm}\begin{table}[H]
     \centering
     \caption{Counterexample team for $\Sigma\not\vdash x\subseteq_{\frac{1}{2}}y$ with $\Sigma=\{x\subseteq_{\frac{1}{4}}w,\, w\subseteq_{\frac{1}{2}}y\}$.}
     \label{R team}\vspace{-.59cm}
 \[
\begin{array}{c ccc cc}
\toprule 	
\quad\,\, &\, x \, &\, w\,&\,  y \,& \dots \\ 
\midrule 

s_1& 1 & 1& 1
  & \dots
  \\
s_2&   2 & 1& 1
  & \dots
  \\
s_3&3 & 3& 1
 & \dots
  \\
 s_4&   4 & 4&1
  & \dots
  \\
s_5&      5& 5& 5
 & \dots
  \\
\bottomrule 
\end{array}
\] \vspace{-.7cm}
 \end{table}


 \begin{theorem}[Complexity] \label{unary complexity}
     Let $\Sigma\cup \{x\subseteq_p y\}$ be a finite set of either unary quantity approximate inclusion atoms or unary ratio approximate inclusion atoms. Deciding whether $\Sigma\models x\subseteq_p y$ is in PTIME. 
 \end{theorem}
\begin{proof}
Use Dijkstra's algorithm with time complexity $O(|Var|^2)$, where $Var$ is the set of all variables appearing in $\Sigma\cup \{x\subseteq_p y\}$. The algorithm solves the single-source shortest path problem for directed graphs with non-negative weights, where the variables correspond to the nodes, the inclusion symbols to the directed edges, and the approximations to the weights.  
\qed\end{proof}
 \section{Conclusion and future work}\label{conclusion}
The paper begins the proof-theoretical study of approximate union-closed atoms by introducing and axiomatizing quantity and ratio approximate inclusion atoms. Complete axiomatizations are provided for quantity approximate inclusion atoms where the assumptions' arities do not exceed that of the conclusion, and for unary ratio approximate inclusion atoms. The decision problems for both unary versions are in PTIME. We identify three directions of future work. 
 
The \emph{$k$-anonymity atom} $x\Upsilon^k y$ is union-closed, and expresses that the value of $x$ appears together with at least $k$ values of $y$. A quantity approximate version could allow $n$ such $x$ values to not satisfy the $k$-anonymity condition. We observe that the usual rules for $k$-anonymity atoms in \cite{Vaananen2022}, augmented with the approximations, are sound. Completeness results and connections to the practical $k$-anonymity algorithms for databases (see, e.g., \cite{Kenig}) are left to be explored.

 Extending the completeness result in \Cref{R compl} beyond \emph{unary} ratio approximate inclusion is a natural pursuit, while expecting an increase in computational complexity. One approach is to develop other counterexample teams to the ones in the literature \cite{CASANOVA198429,haggblom2025}. Another is to modify the definitions, e.g., by forcing \emph{locality} by defining $T\models x\subseteq_p y$ iff $|T[x]|\leq p|T[x]|$, or $T\models x\subseteq_p y$ iff $|T[x]|\leq p|T[xy]|$. However, transitivity is not sound for either: $\{x\subseteq_p z, z\subseteq_q y\}\not\models x\subseteq_r y$ for any $p,q>0$ and $r< 1$. The search for other suitable definitions is therefore ongoing.

The languages considered in this paper of quantity and ratio approximate inclusion atoms could be enriched by connectives and quantifiers to study approximate variants of the union-closed \emph{inclusion logic} from \cite{GALLIANI201268}. Approximation operators could also be applied to first-order formulae, an approach taken in \cite{Durand} for the downward-closed dependence logic, with connections to \emph{multiteams} that allow modelling databases with duplicated rows. 
 
 
 
 
 

\begin{credits}
\subsubsection{\ackname}

The author wants to thank Åsa Hirvonen for many helpful suggestions regarding the paper, and Fan Yang for useful discussions.

\subsubsection{\discintname}
The author has no competing interests to declare that are
relevant to the content of this article.
\end{credits}

 \bibliographystyle{splncs04}
\bibliography{mybibliography}
%




\end{document}